\documentclass[11pt]{article}
\usepackage{fullpage,graphicx,amsmath,amssymb,verbatim,bm,xcolor,textcomp}
\allowdisplaybreaks
\newtheorem{definition}{Definition}

\newtheorem{fact}{Fact}
\newtheorem{theorem}{Theorem}
\newtheorem{corollary}{Corollary}

\newenvironment{proof}{\textbf{Proof:}}{\hfill$\square$}

\newcommand{\cE}{\mathcal{E}}

\newcommand{\cN}{\mathcal{N}}

\newcommand{\cT}{\mathcal{T}}

\newcommand{\C}{\mathbb{C}}
\newcommand{\R}{\mathbb{R}}

\newcommand{\Hmin}{H_{\mathrm{min}}}

\newcommand{\supp}{\mathrm{supp}}

\DeclareMathOperator*{\E}{{\rm {\mathbb E}}\,}
\DeclareMathOperator*{\Tr}{{\rm Tr}\;}
\DeclareMathOperator*{\prob}{{\rm Pr}\;}

\newcommand{\zero}{\leavevmode\hbox{\small l\kern-3.5pt\normalsize0}}
\newcommand{\one}{\leavevmode\hbox{\small1\kern-3.8pt\normalsize1}}

\newcommand{\ket}[1]{| #1 \rangle}
\newcommand{\bra}[1]{\langle #1 |}
\newcommand{\ketbra}[1]{\ket{#1}\bra{#1}}
\newcommand{\braket}[2]{\langle {#1} \ket{#2}}

\begin{document}

\title{{\bf Matrix Chernoff concentration bounds for multipartite
           soft covering and expander walks
}}

\author{Pranab Sen\footnote{
School of Technology and Computer Science, Tata Institute of Fundamental
Research, Mumbai, India. 
Email: {\sf pranab.sen.73@gmail.com}.
}
}

\date{}

\maketitle

\begin{abstract}
We prove Chernoff style exponential concentration bounds for classical
quantum soft covering generalising previous works which gave bounds only 
in expectation. Our first result is an exponential concentration bound 
for fully smooth multipartite classical quantum soft covering, extending 
Ahlswede-Winter's seminal result \cite{Ahlswede:matrixchernoff} in several 
important directions.
Next, we prove a new exponential concentration
result for smooth unipartite classical quantum soft covering when
the samples are taken via a random walk on an expander graph. The resulting
expander matrix Chernoff bound complements the results of Garg, Lee,
Song and Srivastava
\cite{Garg:expandermatrixchernoff} in important ways. We prove our
new expander matrix Chernoff bound by generalising McDiarmid's method
of bounded differences for functions of independent random variables
to a new method of {\em bounded excision for functions of
expander walks}. This new technical tool should be of independent interest.

A notable feature of our new concentration bounds is that they have no
explicit Hilbert space dimension factor. This is because our bounds are
stated in terms of the Schatten $\ell_1$-distance of the sample
averaged quantum state to the `ideal' quantum state. Our bounds are 
sensitive to certain smooth R\'{e}nyi max divergences, giving a clear
handle on the number of samples required to achieve a target 
$\ell_1$-distance. Using these novel features, we prove new one 
shot inner bounds for sending
private classical information over different kinds of quantum wiretap
channels with many non-interacting eavesdroppers that are independent
of the Hilbert space dimensions of the eavesdroppers. Such powerful
results were unknown earlier even in the fully classical  setting.
\end{abstract}

\section{Introduction}
The foundational works of Chernoff \cite{Chernoff:conc} and 
Hoeffding \cite{Hoeffding:conc} showed that an average
of $n$ independent samples of a bounded random variable is concentrated 
around its true mean exponentially in $n$. Their results, and 
extensions thereof,
have found countless applications in probability theory, statistics and
computer science. 

Chernoff's and Hoeffding's results were for sample averages of bounded 
real valued random variables. 
In their seminal paper, Ahlswede and Winter \cite{Ahlswede:matrixchernoff}
extended Chernoff style concentration results to matrix valued random
variables; they called their main result a {\em matrix Chernoff bound}. 
Several later works have improved and extended Ahlswede-Winter's
result in multiple ways. Tropp's book \cite{Tropp:book} contains a
detailed survey of various matrix Chernoff bounds together with several
applications. 

An important area of application of matrix Chernoff bounds is in quantum
information theory. More specifically, matrix Chernoff bounds can be
used to prove concentration results for so-called {\em classical quantum
soft covering} problems. In the basic version of the problem, there is a 
classical
random variable $X$ taking value $x$ with probability $p(x)$. For each
$x$, one is given a quantum state aka density matrix $\rho_x^M$, i.e. a
complex Hermitian positive semidefinite matrix with unit trace, acting on
a Hilbert space $M$. Define the `ideal' mean quantum state
$\rho^M := \E_{x}[\rho_x^M] := \sum_x p(x) \rho_x^M$. We want upper bounds
on the following tail probability:
\[
\Pr_{x_1, \ldots, x_K}
\left[
\left\|
\frac{1}{K} \sum_{i=1}^K \rho_{x_i}^M
- \rho^M
\right\|_1
> \delta
\right]
\]
where $\delta > 0$, $\|\cdot\|_1$ denotes the Schatten $\ell_1$-norm
aka the {\em trace norm} of matrices, and the probability is taken
over $x_1, \ldots, x_K$ with each $x_i$ being independently chosen with
probability $p(x_i)$. The tail upper bound is considered to be Chernoff
style or exponential if it is less than $A \exp(-K B)$ for suitable 
$A$, $B$ that may depend on the ensemble $\{\rho_x^M\}_x$, the dimension 
of $M$ and $\delta$ but cannot depend on the number of samples $K$.
We will call an above such problem a {\em soft classical quantum 
covering problem in concentration}.

Quantum information theory has extensively studied a 
related problem that we call a {\em soft classical quantum 
covering problem in expectation}. In this variant, we are interested
in upper bounds on the expected trace distance between the sample
averaged state and the ideal state in expectation over the choice of
$x_1, \ldots, x_K$ viz. we want an upper bound on
\[
\E_{x_1, \ldots, x_K}
\left[
\left\|
\frac{1}{K} \sum_{i=1}^K \rho_{x_i}^M
- \rho^M
\right\|_1
\right],
\]
where the probability is taken
over $x_1, \ldots, x_K$ with each $x_i$ being independently chosen with
probability $p(x_i)$.

Fully classical versions of the above soft covering questions were defined 
and studied in the work of Cuff \cite{Cuff:softcovering}.
Cuff's soft covering results were proved in the classical asymptotic 
independent and identically distributed (iid) setting, where one
takes samples not from the random variable $X$ but from $X^{\times n}$ i.e.
$n$ independent copies of $X$. Henceforth, the setting where one has to 
take independent samples from a single copy of $X$ will be called the
{\em one shot} setting. One shot setting is the most fundamental setting;
the asymptotic iid setting can be derived from the one shot setting by
treating $X^{\times n}$ as a single copy of a new random variable $Y$.
One shot setting is important in its own right when the iid assumption
is not valid; it also often serves as the starting point in proving
second order finite blocklength results.
Hence, it is important to obtain one shot classical quantum soft covering
results in both expectation as well as concentration. Prior to this work,
we are unaware of any concentration result for soft covering
in the classical one shot setting. 

After Cuff's work, a fully quantum version of soft covering in
expectation was defined and studied in the seminal work of
Anshu, Devabathini and Jain \cite{Jain:convexsplit}. They called it 
the {\em convex split lemma}. For the classical
quantum setting which is intermediate between the fully classical and
fully quantum settings, the convex split lemma reduces to the 
classical quantum soft covering problem in expectation given above.
The classical quantum version arises in the study of so called 
{\em covering style} problems for handling classical messages via
quantum channels.
Most often, one only needs covering in expectation. For example, inner
bounds for private classical communication over quantum wiretap 
channels \cite{Devetak:wiretap, Wilde:wiretap}, 
though sometimes published using
the concentration version, actually require only the expectation version of
classical quantum soft covering.
A similar remark can be made for measurement compression results
\cite{Winter:measurementcompression}. Classical quantum soft covering in 
concentration is 
required in specialised cases e.g. to prove good dimension independent 
inner bounds for wiretap channels having many
{\em non-interacting eavesdroppers}, as we will see in detail below.

A covering result in concentration clearly implies a result in expectation.
The converse is generally false. Until now, soft covering results in 
expectation and concentration were proved using completely different sets 
of techniques. To the best of our knowledge, no prior work proved a 
concentration result by first proving an expectation result. This
shortcoming becomes even more evident when one aims for 
{\em multipartite classical quantum soft covering} results. In its 
simplest bipartite version, there
are two independent random variables $X$ and $Y$ and a classical quantum
mapping $(x,y) \mapsto \rho_{xy}^M$. The `ideal' mean quantum state
is defined as 
$\rho^M := \E_{x,y}[\rho^M_{xy}] := \sum_{x,y} p(x) p(y) \rho^M_{x,y}$.
The bipartite soft covering lemmas in expectation and concentration 
seek upper bounds on the following quantities respectively:
\[
\E_{\substack{x_1, \ldots, x_K \\ y_1, \ldots, y_L}}
\left[
\left\|
\frac{1}{KL} \sum_{i=1}^K \sum_{j=1}^L \rho_{x_i,y_j}^M
- \rho^M
\right\|_1
\right]
~~~ \mbox{and} ~~~
\Pr_{\substack{x_1, \ldots, x_K \\ y_1, \ldots, y_L}}
\left[
\left\|
\frac{1}{KL} \sum_{i=1}^K \sum_{j=1}^L \rho_{x_i,y_j}^M
- \rho^M
\right\|_1
> \delta
\right],
\]
where the expectations are taken over independent choices of
$x_1, \ldots, x_K$ from $X^{\times K}$, and independent choices of
$y_1, \ldots, y_L$ from $Y^{\times K}$. In particular, the choices of
$x_1, \ldots, x_K$ are also independent from the choices of
$y_1, \ldots, y_L$.

Till the work of Anshu, Jain and Warsi \cite{anshu:slepianwolf}, no
multipartite soft covering lemma in expectation was known. Prior to
the present work, no multipartite soft covering lemma in concentration was 
known either. Multipartite covering lemmas are natural tools for tackling
multiterminal versions of covering problems in information theory.
For example, good inner bounds for sending private classical information
over wiretap quantum multiple
access channels (wiretap QMAC) \cite{Chakraborty:wiretapQMAC}, or good 
achievabilty results for the
{\em centralised multilink measurement compression} problem 
\cite{Chakraborty:measurementcompression}
can be proved if one were to have a powerful multipartite soft covering 
lemma in expectation. Similarly,  a powerful multipartite soft covering 
lemma in concentration would allow one to prove dimension independent 
inner bounds for a wiretap QMAC with many non-interacting eavesdroppers,
as will become clearer below.

The statements of the unipartite one shot soft covering lemma in 
expectation
proved earlier were given in terms of a smooth one shot R\'{e}nyi max 
divergence property of the ensemble of quantum states. Bounds in terms 
of smooth one shot divergences are much more desirable than bounds in 
terms of non-smooth one shot quantities. This is because only the 
smooth one shot version
converges to the correct quantity in the asymptotic iid limit. Moreover,
only the smooth one shot version leads to good second order inner bounds.
Since the unipartite soft covering
lemmas in concentration known earlier were proved using different
techniques, they did not contain any explicit dependence of the 
concentraion
on the smooth R\'{e}nyi max divergence property.  This lacuna was first
addressed in the work of Radhakrishan, Sen and Warsi 
\cite{Radhakrishnan:wiretap}, whose statement of their unipartite soft 
covering lemma in concentration contained the same smooth R\'{e}nyi max 
divergence property used in the statement of the expectation version.
However, even that work proved the concentration result directly, and
did not make use of the techniques used to show the expectation version.
Radhakrishnan, Sen and Warsi also proved a new exponential Chernoff style 
concentration result for non-square matrices in terms of a similar
smooth R\'{e}nyi max divergence property. However, both of their results
have an explicit Hilbert space
dimension factor which makes their inner bound for the quantum wiretap
channel dependent on the dimensions of the Hilbert spaces of the 
eavesdroppers. More specifically, Radhakrishnan, Sen and Warsi's inner 
bound takes an additive hit
of $\log \log D$, where $D$ is the largest Hilbert space dimension of an
eavesdropper. This is unsatisfactory.

Till recently, the only results known for multipartite soft classical
quantum covering in expectation uere stated in terms of non-smooth
R\'{e}nyi divergence quantities. The works \cite{Sen:telescoping,
Sen:flatten} proved fully smooth multipartite classical quantum soft
covering results in expectation for the first time. Fully smooth 
means that the upper bounds are stated in terms of smooth R\'{e}nyi
divergence quantities of all possible subsets of the classical random 
variables involved in the soft covering. The work of
\cite{Sen:telescoping} proved fully smooth multipartite soft covering 
bounds in expectation,
 as well as fully smooth multipartite convex split bounds in terms of
smooth R\'{e}nyi $2$-divergence quantities. This automatically implies
the same bounds in terms of smooth R\'{e}nyi max divergence quantities.
Moreover, the soft covering bounds continue to hold even if 
$x_1, \ldots, x_K$ are chosen in pairwise independent fashion with each
marginal having the distribution of $X$, and similarly for 
$y_1, \ldots, y_L$, with $x_1, \ldots, x_K$ being independent of
$y_1, \ldots, y_L$.  The work of \cite{Sen:flatten} showed that the same
fully smooth upper bound in terms of smmoth R\'{e}nyi max divergence 
continues to hold even if the pairwise independent assumption is slightly
weakened. Since \cite{Sen:telescoping, Sen:flatten} show that 
fully smooth
multipartite soft covering in expectation holds even under weak 
assumptions like `almost' pairwise independence, one wonders if 
multipartite soft convering in concentration would hold if one had
stronger assumptions like full independence of the samples.
Unfortunately, the above two works did not prove multipartite soft
covering results in concentration. The present paper fulfills this
shortcoming.

A different generalisation of the Chernoff bound was given by
Gillman \cite{Gillman:expanderchernoff}, who showed exponential 
concentration of the sample average around the true mean when the
samples are taken via a random walk on an expander graph. Since sampling
from an expander walk requires much less randomness than sampling 
independently, Gillman's {\em expander Chernoff} result immediately
finds several applications towards randomness efficient sampling and
derandomisation \cite{Gillman:expanderchernoff}. 
Using sophisticated techniques, Garg, Lee, Song and 
Srivastava \cite{Garg:expandermatrixchernoff} managed to marry 
Ahlswede-Winter's matrix Chernoff bound with Gillman's expander Chernoff 
bound obtaining, for the first time, an {\em expander matrix Chernoff
bound} as follows:
\begin{quote}{\bf \cite[Theorem 1.2]{Garg:expandermatrixchernoff}}
Let $X$ be the vertex set of a regular, undirected constant degree 
expander graph with second eigenvalue of absolute value $\lambda$.
Let $f: X \rightarrow \C^{d \times d}$ be a $d \times d$ matrix valued
function on $X$ such that the Schatten $\ell_\infty$-norm 
$\|f(x)\|_\infty \leq 1$ for all $x \in X$ and
$\sum_{x \in X} f(x) = 0$. Let $x_1, \ldots, x_K$ be a 
{\em stationary random walk on $X$} i.e. $x_1$ is chosen from the unique
stationary distribution on $X$ which happens to be the uniform 
distribution,
and then $x_2, x_3, \ldots,$ are chosen via a random walk starting from
$x_1$. Let $0 < \epsilon < 1$. Then:
\[
\Pr_{x_1, \ldots, x_K} 
\left[
\left\|
\frac{1}{K} \sum_{i=1}^K f(x_i)
\right\|_\infty 
\geq \epsilon
\right]
\leq
d \exp(-C \epsilon^2 (1-\lambda) K),
\]
where $C$ is a universal constant.
\end{quote}
Again, such an expander matrix Chernoff bound has several applications
in derandomisation \cite{Xiao:expandermatrixchernoff}.

Our first result viz. Theorem~\ref{thm:smoothcoveringconc} below is 
a fully smooth
multipartite classical quantum soft covering lemma in 
concentration.  This result leads, to the best of our knowledge, 
the first matrix Chernoff
bound where, to get a target distance between the sample averaged state
and the ideal state, the sample size only depends on a smooth R\'{e}nyi max
divergence term
and the concentration bound has no dependence on the dimension of the
ambient Hilbert space of the matrices. 
We believe that this is an important conceptual contribution of this work.
Even restricted to the unipartite 
setting as in Corollary~\ref{cor:matrixchernoff}, our new matrix 
Chernoff bound generalises the Ahlswede-Winter 
bound in these two senses. This is  because the Ahlswede-Winter bound
was stated in terms of a non-smooth max
divergence; also it had an additional  weak dependence on 
the dimension of the
Hilbert space $M$ of the density matrices viz. their concentration result
was of the form $|M| \exp(-|A| \cdots)$.
These two deficiencies are 
significant weaknesses in some applications  to 
quantum information theory e.g. in proving inner bounds for the quantum
wiretap channel independent of the dimensions of the Hilbert spaces of
the eavesdroppers.
Intuitively, these two improvements become possible because our new matrix
Chernoff bound guarantees closeness in the Schatten $\ell_1$-distance 
whereas the Ahlswede-Winter bound guarantees closeness in the 
Schatten $\ell_\infty$-distance. We prove our multipartite soft covering
lemma in concentration by taking the fully smooth multipartite soft
covering lemma in expectation of \cite{Sen:telescoping} and then applying
McDiarmid's {\em method of bounded differences} for independent random 
variables \cite{McDiarmid:bounded} to upper bound the tail
probability.

As a consequence of our new matrix Chernoff bound, we obtain 
the first inner bound in Theorem~\ref{thm:wiretap} below for sending
private classical information over a point to point quantum wiretap
channel in the presence of many non-interacting eavesdroppers that
does not depend on the dimensions of the Hilbert spaces of the
eavesdroppers. As mentioned above,
\cite{Radhakrishnan:wiretap} proved a weak dimension dependent inner
bound for a wiretap channel with many eavesdroppers using their
weak dimension dependent matrix Chernoff bound.
The paper \cite{Wilde:wiretap} proved a dimension independent inner 
bound for one eavesdropper; it could not handle many eavesdroppers
because it used a smooth covering lemma in expectation, not concentration.
Earlier works on both classical and quantum wiretap channels 
used other techniques
but nevertheless, they too could not handle many eavesdroppers. Thus,
Theorem~\ref{thm:wiretap} is a new result even for classical
asymptotic iid information theory. An analogous new 
eavesdropper-dimension independent inner bound for
private classical communication over a wiretap QMAC with many 
non-interacting eavesdroppers is proved in 
Theorem~\ref{thm:wiretapQMACconc} below.

Our second result is a unipartite smooth soft covering lemma in 
concentration (Theorem~\ref{thm:expanderconc}) where the samples are 
taken via a random walk on an
expander graph. Put
differently, our second result is a new expander matrix Chernoff bound 
for quantum states where the distance between the sample averaged
state and the ideal state is measured in terms of the Schatten 
$\ell_1$-norm instead of the Schatten $\ell_\infty$-norm used 
by \cite{Garg:expandermatrixchernoff}. Our second result does
not have any explicit dependence on the dimension of the ambient Hilbert
space of the quantum states, and the sample size depends only on the
smooth R\'{e}nyi max divergence of the ensemble and the properties of
the expander graph. Our expander matrix Chernoff bound is obtained by
proving, for the first time, a unipartite smooth soft covering lemma 
in expectation for expander walks using the so-called 
{\em matrix weighted Cauchy Schwarz inequality}. We then develop a 
novel concentration technique which we call the 
{\em method of bounded excision}. This method, which should be of 
independent interest, can be thought of as a 
generalisation of McDiarmid's method of bounded differences where the
independent random samples are replaced by samples from an expander walk.
We then convert the soft covering result in expectation to a
result in concentration by using the method of bounded excision
to prove an upper bound on the tail probability.

Finally, we  note how Garg, Lee, Song and Srivastava's 
\cite{Garg:expandermatrixchernoff} expander matrix Chernoff bound
in terms of Schatten $\ell_\infty$-distance does not give useful
results if used to obtain bounds in terms of Schatten $\ell_1$-distance 
as required in information theoretic applications like the wiretap
channel. If a target maximum distance of $\epsilon$ is required in
Schatten $\ell_1$, one needs to start of with $\epsilon/d$ distance in
Schatten $\ell_\infty$, where $d$ is the dimension of the ambient 
Hilbert space of the quantum states. Then the upper bound provided
by Garg et al.'s bound is
$
d \exp\left(- \frac{C \epsilon^2 (1 - \lambda) K}{d^2}\right),
$
which means that the number of samples $K$ must be at least 
$d^2 \log d$ in order to get even a constant amount of concentration
probability. On the other hand for a constant amount of concentration
probability,  Theorem~\ref{thm:expanderconc} below
gives a bound for $K$ 
which is always at most, and often significantly less than, $d \log |X|$,
where $X$ is the vertex set of the expander graph.

\section{Preliminaries}
\label{sec:prelims}
Let $p \geq 1$.
For any matrix $M$, its {\em Schatten $\ell_p$-norm} is defined as
$\|M\|_p := (\Tr [(M^\dag M)^{p/2}])^{1/p}$. In other words, 
$\|M\|_p$ is the $\ell_p$-norm of the vector of its singular values.
The norm $\|M\|_1$ is also called the {\em trace norm} of $M$.
The norm $\|M\|_2$ is also called the {\em Frobenius norm} or the
{\em Hilbert-Schmidt} norm of $M$ and is nothing but the $\ell_2$-norm 
of the vector that would arise if all the matrix entries of $M$ 
were written down in a linear fashion.
The norm $\|M\|_\infty$ is also called the {\em operator norm} of $M$ 
because of the equality 
$\|M\|_\infty = \max_{\|v\|_2 = 1} \|M v\|_2$.
The support of a Hermitian matrix $M$, denoted by $\supp(M)$, is the 
span of the non-zero eigenspaces of $M$.
For a Hermitian matrix $M$ and $\alpha > 0$, we define $M^\alpha$ in
the natural fashion by appealing to the eigenbasis of $M$. If
$\alpha < 0$, we define $M^\alpha$ in the natural fashion only on 
$\supp(M)$, and keep the zero eigenspace of $M$ as the zero eigenspace 
of $M^\alpha$. Equivalently, we take the inverse $M^{-1}$ on $\supp(M)$
only, a process called the {\em Moore-Penrose pseudoinverse}, and then
define $M^\alpha := (M^{-1})^{-\alpha}$ for $\alpha < 0$.
Similarly, for a positive semidefinite $M$, we define $\log M$ in the
natural fashion on $\supp(M)$ by appealing to the eigenbasis of $M$,
keeping the zero eigenspace of $M$ as the zero eigenspace of $\log M$.

We now define a few entropic quantities that will be required
to state our results below. For a classical alphabet $X$ or a Hilbert
space $X$, $|X|$ denotes the cardinality of the alphabet or the 
dimension of the Hilbert space respectively. The quantity 
$\frac{\one^X}{|X|}$ denotes the uniform probability distribution / 
completely mixed state on  classical alphabet $X$ / Hilbert space $X$
respectively.
\begin{definition}[{\bf (Smooth hypothesis testing divergence)}]
Let $0 < \epsilon < 1$.
Let $\alpha^A$, $\beta^A$
between two density matrices 
defined on the same Hilbert space $A$. The {\em smooth hypothesis testing
divergence of $\alpha^A$ with respect to $\beta^A$} is defined as
\[
D^\epsilon_H(\alpha \| \beta) :=
\max_{\Pi: \zero^A \leq \Pi \leq \one^A}
\{
-\log \Tr[\Pi \beta] : \Tr[\Pi \alpha] \geq 1 - \epsilon
\},
\]
where the maximisation is over all POVM elements $\Pi$ on $A$.
\end{definition}
\begin{definition}[{\bf (Smooth hypothesis testing mutual information)}]
Let $0 < \epsilon < 1$.
The {\em smooth hypothesis testing mutual information under joint quantum 
state $\rho^{AB}$} is defined as
\[
I^\epsilon_H(A:B)_\rho :=
D^\epsilon_H(\rho^{AB} \| \rho^A \otimes \rho^B).
\]
\end{definition}
\begin{definition}[{\bf (Smooth conditional hypothesis testing mutual 
information)}]
Let $0 < \epsilon < 1$.
Let $Q$, $X$, $Y$ be classical systems. Let $C$ be a quantum system. 
Let $p^{QXY}$ be a normalised joint
probability distribution on $QXY$ of the form $p(q) p(x|q) p(y|q)$
i.e. $X$ and $Y$ are independent given $Q$. Let 
$(x,y) \mapsto \sigma^C_{xy}$ be a classical to quantum mapping.
The joint classical quantum state $\sigma^{QXYC}$ is defined as
\[
\sigma^{QXYC} :=
\sum_{qxy} p(q)p(x|q) p(y|q) \ketbra{q,x,y}^{QXY} \otimes \sigma^C_{xy}.
\]
The {\em smooth conditional hypothesis testing mutual information under
$\sigma^{QXYC}$} is defined as
\begin{eqnarray*}
I^\epsilon_H(XY : C | Q)_\sigma 
& := &
D^\epsilon_H(\sigma^{QXYC} \| \bar{\sigma}^{XY:C|Q}), \\
\bar{\sigma}^{XY:C|Q} 
& := &
\sum_{qxy} p(q)p(x|q) p(y|q) \ketbra{q,x,y}^{QXY} \otimes
\sigma^C_q, \\
\sigma^C_q 
& := &
\sum_{x,y} p(x|q) p(y|q) \sigma^C_{xy}.
\end{eqnarray*}
\end{definition}
\begin{definition}[{\bf (R\'{e}nyi divergences)}]
Let $M$ be a quantum system. Let $\alpha^M$ be a subnormalised quantum
state and $\beta^M$ be a normalised quantum state in $M$. The 
{\em (non-smooth) R\'{e}nyi $2$-divergence} of $\alpha^M$ with 
respect to $\beta^M$ is defined as
\[
D_2(\alpha \| \beta) := 
2 \log\|\beta^{-1/4} \alpha \beta^{-1/4}\|_2, ~~ 
\mbox{if $\supp(\alpha) \leq \supp(\beta)$, $+\infty$ otherwise}.
\]
The {\em (non-smooth) R\'{e}nyi $\infty$-divergence} aka 
{\em (non-smooth) R\'{e}nyi max divergence}
of $\alpha^M$ with respect to $\beta^M$ is defined as
\[
D_\infty(\alpha \| \beta) := 
\log\|\beta^{-1/2} \alpha \beta^{-1/2}\|_\infty, ~~
\mbox{if $\supp(\alpha) \leq \supp(\beta)$, $+\infty$ otherwise}.
\]
\end{definition}
\begin{definition}[{\bf (Shannon divergence)}]
The (non-smooth) {\em Shannon divergence} aka 
{\em Kullback-Leibler divergence} 
aka {\em relative entropy} of $\alpha^M$ with 
respect to $\beta^M$ is defined as
\[
D(\alpha \| \beta) := 
\Tr [\alpha (\log \alpha - \log \beta)], ~~
\mbox{if $\supp(\alpha) \leq \supp(\beta)$, $+\infty$ otherwise}.
\]
\end{definition}
\begin{definition}[{\bf (Smooth R\'{e}nyi divergences)}]
Let $0 < \epsilon < 1$.
The {\em $\epsilon$-smooth R\'{e}nyi divergences} are defined as follows:
\[
D_2^\epsilon(\alpha \| \beta) 
:=  
\min_{\alpha' \approx_\epsilon \alpha}
D_2(\alpha' \| \beta), ~~
D_\infty^\epsilon(\alpha \| \beta) 
:= 
\min_{\alpha' \approx_\epsilon \alpha}
D_\infty(\alpha' \| \beta), 
\]
where the minimisation is over all subnormalised density matrices 
$\alpha'$ satisfying $\|\alpha' - \alpha\|_1 \leq \epsilon (\Tr\alpha)$. 
\end{definition}
\begin{definition}[{\bf (R\'{e}nyi and Shannon mutual information)}]
Let $\sigma^{XE}$ be a joint quantum state on the quantum system $X E$.
The {\em (non-smooth) R\'{e}nyi-2 and R\'{e}nyi-max mutual informations} 
are defined as
\[
\begin{array}{c}
I_2(X:E)_\sigma := 
D_2(\sigma^{XE} \| \sigma^{X} \otimes \sigma^E), \\
I_\infty(X:E)_\sigma := 
D_\infty(\sigma^{XE} \| \sigma^{X} \otimes \sigma^E), \\
\end{array}
\]
The {\em smooth R\'{e}nyi-2 and R\'{e}nyi-max mutual informations} 
are defined as
\[
I_2^\epsilon(X : E)_\sigma
:=  
D_2^\epsilon(\sigma^{XE} \| \sigma^{X} \otimes \sigma^{E}), \\
I_\infty^\epsilon(X : E)_\sigma
:= 
D_\infty^\epsilon(\sigma^{XE} \| \sigma^{X} \otimes \sigma^{E}).
\]
The (non-smooth) Shannon mutual information is also defined similarly.
\[
I(X:E)_\sigma := 
D(\sigma^{XE} \| \sigma^{X} \otimes \sigma^E). 
\]
\end{definition}
\begin{definition}[{\bf (R\'{e}nyi and Shannon conditional entropies)}]
The {\em (non-smooth) R\'{e}nyi-2 and R\'{e}nyi-max conditional entropies} 
are defined as
\[
\begin{array}{c}
H_2(X|E)_\sigma := 
\log |X| - D_2(\sigma^{XE} \| \frac{1^X}{|X|} \otimes \sigma^{E}), \\
\Hmin(X|E)_\sigma := H_\infty(X|E)_\sigma := 
\log |X| - D_\infty(\sigma^{XE} \| \frac{1^X}{|X|} \otimes \sigma^E).
\end{array}
\]
The (non-smooth) {\em Shannon conditional entropy} is defined similarly.
\[
H(X|E)_\sigma := 
\log |X| - D(\sigma^{XE} \| \frac{1^X}{|X|} \otimes \sigma^E).
\]
The {\em smooth R\'{e}nyi-2 and R\'{e}nyi-max conditional entropies} 
are defined as
\[
\begin{array}{c}
H_2^\epsilon(X|E)_\sigma := 
\min_{\sigma^{'XE} \approx_\epsilon \sigma^{XE}}
\{\log |X| - D_2(\sigma^{'XE} \| \frac{1^X}{|X|} \otimes \sigma^{'E})\}, \\
\Hmin^\epsilon(X|E)_\sigma := H^\epsilon_\infty(X|E)_\sigma := 
\min_{\sigma^{'XE} \approx_\epsilon \sigma^{XE}}
\{\log |X| - D_\infty(\sigma^{'XE}\|\frac{1^X}{|X|} \otimes \sigma^{'E})\}.
\end{array}
\]
Above, 
the minimisation is over all subnormalised density matrices 
$\sigma^{'XE}$ satisfying 
$\|\sigma^{'XE} - \sigma^{XE}\|_1 \leq \epsilon (\Tr\,\sigma^{XE})$. 
\end{definition}

For the entropic quantities above, we note that
\[
\begin{array}{c}
D(\alpha \| \beta) \leq D_2(\alpha \| \beta) \leq 
D_\infty(\alpha \| \beta) \\
{} \implies 
I(X:E)_\sigma \leq I_2(X:E)_\sigma \leq I_\infty(X:E)_\sigma, ~~
D_2^\epsilon(\alpha \| \beta) \leq D_\infty^\epsilon(\alpha \| \beta), ~~
I_2^\epsilon(X:E)_\sigma \leq I_\infty^\epsilon(X:E)_\sigma.
\end{array}|
\]
Observe that the mutual information and conditional entropy quantities
are always finite. Slight variants of the above definitions for smooth
entropic quantities are also available in the literature, but they
are all `roughly' equivalent.
Note also that for a classical quantum (cq) state $\sigma^{XE}$, 
\[
I_2(X:E)_\sigma =
\log \E_{x} \Tr[((\sigma^E)^{-1/4} \sigma^E_{x} (\sigma^E)^{-1/4})^2],
\]
the expecation over $x$ being taken according to the classical probability
distribution $\sigma^{X}$.
For later use, we remark that for a cq state $\sigma^{XE}$,
\[
I_\infty(X:E)_\sigma =
\max_{x} D_\infty(\sigma^E_{x} \| \sigma^E).
\]

We will use the following {\em matrix weighted Cauchy-Schwarz} inequality
below in order to prove upper bounds on the trace norm.
\begin{fact}
\label{fact:matrixCauchySchwarz}
Let $M$ be a Hermitian matrix and $\sigma$ be a normalised density
matrix on the same Hilbert space. Suppose $\supp(M) \leq \supp(\sigma)$.
The matrix $\sigma$ is called a {\em weighting matrix}.
Then,
\[
\|M\|_1 \leq \|\sigma^{-1/4} M \sigma^{-1/4}\|_2.
\]
\end{fact}

We need
McDiarmid's probability concentration inequality aka 
{\em McDiarmid's method of bounded differences} \cite{McDiarmid:bounded}.
\begin{fact}
\label{fact:mcdiarmid}
Let $k$ be a positive integer and $X_1, \ldots, X_k$ be $k$ classical 
alphabets. Let $f: X_{[k]} \rightarrow \R$ be a function satisfying
the following {\em bounded differences property} for positive reals
$c_1, \ldots, c_k$:
\[
\forall (x_1, \ldots, x_k) \in X_{[k]}:
\forall i \in [k]:
\forall x'_i \in X_i:
~~
|f(x_1, \ldots, x_i, \ldots, x_k) - f(x_1, \ldots, x'_i, \ldots, x_k)|
\leq c_i.
\]
Define $c^2 := c_1^2 + \cdots + c_k^2$.
Let $\delta > 0$.
Consider independent probability distributions $q^{X_i}$, $i \in [k]$.
Let 
$
\cE := \E_{x_1, \ldots, x_k}[f(x_1, \ldots, x_k)],
$
where the expectation is taken under these independent distributions.
Then, again under these independent distributions,
\[
\prob_{x_1, \ldots, x_k}\left[
f(x_1, \ldots, x_k) \geq \cE + \delta
\right] \leq
\exp\left(-\frac{2 \delta^2}{c^2}\right).
\]
\end{fact}

Next, we need the 
{\em fully smooth multipartite soft classical quantum covering lemma
in expectation} from \cite{Sen:telescoping}.
\begin{fact}
\label{fact:smoothcoveringexpectation}
Let $k$ be a positive integer. Let $X_1, \ldots, X_k$ be $k$ classical
alphabets. For any subset
$S \subseteq [k]$, let $X_S := (X_s)_{s \in S}$.
Let $p^{X_{[k]}}$ be a normalised probability distribution on $X_{[k]}$.
The notation $p^{X_S}$ denotes the marginal distribution on $X_S$. 
Let $q^{X_1}, \ldots, q^{X_k}$ be normalised
probability distributions on the respective alphabets.
For each $(x_1, \ldots, x_k) \in X_{[k]}$, 
let $\rho^M_{x_1, \ldots, x_k}$
be a subnormalised density matrix on $M$. 
The classical quantum {\em control state} is now defined as
\[
\rho^{X_{[k]} M} :=
\sum_{(x_1, \ldots, x_k) \in X_{[k]}}
p^{X_{[k]}}(x_1, \ldots, x_k) \ketbra{x_1, \ldots, x_k}^{X_{[k]}}
\otimes \rho^M_{x_1, \ldots, x_k}.
\]
Suppose $\supp(p^{X_i}) \leq \supp(q^{X_i})$.
For any subset $S \subseteq [k]$, let 
$q^{X_S} := \times_{s \in S} q^{X_s}$.
Let $A_1, \ldots, A_k$ be positive integers.
For each $i \in [k]$, let
$x_i^{(A_i)} := (x_i(1), \ldots, x_i(A_i))$ denote a $|A_i|$-tuple
of elements from  $X_i$.
Denote the $A_i$-fold product alphabet 
$X_i^{A_i} := X_i^{\times A_i}$, and the product probability distribution
$
q^{X_i^{A_i}} :=
(q^{X_i})^{\times A_i}.
$
For any collection of tuples
$x_i^{(A_i)} \in X_i^{A_i}$, $i \in [k]$, we
define the {\em sample average covering state}
\[
\sigma^M_{x_1^{(A_1}, \ldots, x_k^{(A_k)}} :=
(A_1 \cdots A_k)^{-1}
\sum_{a_1 = 1}^{A_1} \cdots \sum_{a_k = 1}^{A_k}
\frac{p^{X_{[k]}}(x_1(a_1), \ldots, x_k(a_k))}
     {q^{X_1}(x_1(a_1)) \cdots q^{X_k}(x_k(a_k))}
\rho^M_{x_1(a_1), \ldots, x_k(a_k)},
\]
where the fraction term above represents the `change of measure' from
the product probability distribution $q^{X_{[k]}}$ to the
joint probability distribution $p^{X_{[k]}}$.
Suppose for each non-empty subset $\{\} \neq S \subseteq [k]$,
\begin{eqnarray*}
\sum_{s \in S} \log A_s 
& > &
D_2^\epsilon(\rho^{X_S M} \| q^{X_S} \otimes \rho^M) + 
\log \epsilon^{-2}.
\end{eqnarray*}
Then,
\[
\E_{x_1^{(A_1)}, \ldots, x_k^{(A_k)}}[
\|\sigma^M_{x_1^{(A_1)}, \ldots, x_k^{(A_k)}} - \rho^M\|_1
] < 
2(3^k - 1) \epsilon (\Tr \rho).
\]
where the expectation is taken over independent choices of tuples
$x_i^{(A_i)}$ from the distributions $q^{X_i^{A_i}}$, $i \in [k]$.
\end{fact}

We now recall Hoeffding's lemma from probability theory. A proof
can be found for example in \cite{Ying:bounded}.
\begin{fact}
\label{fact:HoeffdingLemma}
Suppose a real valued random variable $X$ satisfies $\E[X] = 0$ and
$a \leq X \leq b$ almost surely. Then for all $h > 0$,
\[
\E[e^{h X}] \leq \exp\left(\frac{h^2 (b - a)^2}{8}\right).
\]
\end{fact}

Finally, we need the following standard property about random walks on an
expander graph. A proof can be inferred, for example, from the calculations
in \cite[Theorem 6.21]{MotwaniRaghavan}.
\begin{fact}
\label{fact:expanderwalk}
Let $G$ be a constant degree undirected expander graph with vertex set $X$.
Let the second largest eigenvalue in absolute value of the transition
matrix of the random walk on $G$ have absolute value $\lambda$. Let
$p(0)^X$ be an initial probability distribution on the vertex set $X$.
Let $p(t)^X$ be the probability distribution on $X$ arising after a 
$t$-step random walk on $G$ starting from distribution $p(0)^X$. Then,
\[
p(t)^X = \frac{1^X}{|X|} + q(t)^X,
\]
where $q(t)^X$ is a vector on $X$ with real entries such that
$
\braket{1^X}{q(t)^X} = 0, \|q(t)\|_2 \leq \lambda^t.
$
\end{fact}
For stating the smooth covering lemmas in expectation and concentration,
we repeat the definition of a {\em stationary expander walk} as 
follows:
\begin{definition}[{\bf (Statonary expander walk)}]
Let $G$ be a constant degree undirected expander graph with vertex set $X$.
A {\em stationary expander walk} of length $K$ on $G$ is a sequence of
$K$ vertices $x_1, \ldots, x_K$  of $G$ where $x_1$ is chosen from the
uniform, which is also the stationary, distribution on $X$ and then
$x_2, \ldots, x_K$ are chosen via a random walk on $G$ starting from $x_1$.
\end{definition}

\section{Fully smooth multipartite soft covering in concentration}
We can now prove our fully smooth multipartite classical quantum soft
covering lemma in concentration.
\begin{theorem}
\label{thm:smoothcoveringconc}
Under the setting of Fact~\ref{fact:smoothcoveringexpectation},
\[
\prob_{x_1^{(A_1)}, \ldots, x_k^{(A_k)}}\left[
\|\sigma^M_{x_1^{(A_1)}, \ldots, x_k^{(A_k)}} - \rho^M\|_1
> 2(3^k - 1) \epsilon (\Tr\rho) + \delta
\right] < 
\exp\left(-\frac{\bar{A} \delta^2}{2 k (\Tr\rho)^2}\right),
\]
where the probability is taken over independent choices of tuples
$x_i^{(A_i)}$ from the distributions $q^{X_i^{A_i}}$, $i \in [k]$,
and $\bar{A}$ is the harmonic mean of $A_1, \ldots, A_k$ defined as
$\bar{A}^{-1} := k^{-1} (|A_1|^{-1} + \cdots + |A_k|^{-1})$.
\end{theorem}
\begin{proof}
We apply Fact~\ref{fact:mcdiarmid} with $(A_1 + \cdots + A_k)$ many
alphabets $X_1^{A_1}, \ldots, X_k^{A_k}$, and function 
$f: X_1^{A_1} \times \cdots \times X_k^{A_k} \rightarrow \R$
defined by
\[
f(x_1^{(A_1)}, \ldots, x_k^{(A_k)}) :=
\|\sigma^M_{x_1^{(A_1)}, \ldots, x_k^{(A_k)}} - \rho^M\|_1.
\]
There will be $(A_1 + \cdots + A_k)$ many bounded differences which 
we will denote by 
$
\{c_i(a_i)\}_{i \in [k], a_i \in [A_i]}.
$
It is easy to see that for any $i \in [k]$, $c_i \in [A_i]$,
$
c_i(a_i) \leq \frac{2 (\Tr\rho)}{A_i}.
$
Then, the quantity $c$ in Fact~\ref{fact:mcdiarmid} becomes
\[
c^2 = \sum_{i=1}^k \sum_{a_i=1}^{A_i} c_i(a_i)^2 = 
4 k \bar{A}^{-1} (\Tr\rho)^2.
\]
The theorem now follows from Fact~\ref{fact:mcdiarmid} and 
Fact~\ref{fact:smoothcoveringexpectation}.
\end{proof}

For $k = 1$, Theorem~\ref{thm:smoothcoveringconc} leads to the following
corollary which can be thought of
as a new {\em matrix Chernoff bound in terms of $D^\epsilon_2$}, in
the sense that the sample size $A \equiv A_1$ has to be lower bounded
by an expression involving $D^\epsilon_2$. Moreover, our bound has no
explicit dimension dependence. In these two senses, it generalises
the original Ahlswede-Winter matrix Chernoff bound.
\begin{corollary}
\label{cor:matrixchernoff}
Let $X$ be a classical alphabet with a normalised probability distribution
$p^X$ on it. For every $x \in X$, let $\rho_x^M$ be a normalised 
density matrix on the Hilbert space $M$. Define the classical quantum
state 
\[
\rho^{XM} := \sum_{x \in  X} p(x) \ketbra{x}^X \otimes \rho_x^M.
\]
Let $A$ be a positive integer.
For a tuple  $(x(1), \ldots, x(A)) \in X^A$, define the sample average
state
\[
\sigma^M_{x(1), \ldots, x(A)} :=
A^{-1} \sum_{a=1}^A \rho^M_{x(a)}.
\]
Let $0 < \epsilon < 1$ and $\delta > 0$.
Suppose
\[
\log A > 
D^\epsilon_2(\rho^{XM} \| p^X \otimes \rho^M) + \log \epsilon^{-2} = 
I^\epsilon_2(X:M)_\rho + \log \epsilon^{-2}.
\]
Then,
\[
\prob_{x(1), \ldots, x(A)}\left[
\|\sigma^M_{x(1), \ldots, x(A)} - \rho^M\|_1 >
3 \epsilon + \delta
\right] <
\exp\left(-\frac{A \delta^2}{2}\right),
\]
where the probability is taken over the choice of
$(x(1), \ldots, x(A))$ from the iid distribution
$(p^X)^{\times A}$.
\end{corollary}

Corollary~\ref{cor:matrixchernoff} can now be used to prove our 
new eavesdropper dimension independent inner 
bound for sending private classical information over a quantum 
wiretap channel with many non-interacting eavesdroppers. The proof 
technique is the same as in \cite{Radhakrishnan:wiretap}.
\begin{theorem}
\label{thm:wiretap}
Let $\cT^{A \rightarrow B E_1 \cdots E_t}$ be a point to point quantum
wiretap channel (completely positive trace preserving (CPTP) 
superoperator) from 
sender $A$ to legitimate receiver $B$ with 
non-interacting eavesdroppers $E_1, \ldots, E_t$. Let $X$ be a 
classical alphabet. Fix a `control' normalised probability distribution
$p^X$ on $X$. Fix a classical to quantum encoding $x \mapsto \rho^A_x$
where $\rho^A_x$ is a normalised quantum state on $A$. Define the
classical quantum `control state'
\[
\rho^{X B E_1 \cdots E_t} :=
\sum_{x\in X} p(x) 
\ketbra{x}^X \otimes 
\cT^{A \rightarrow B E_1 \cdots E_t}(\rho^A_x).
\]
Let $0 < \epsilon < 1$.
Let
\[
R < I^\epsilon_H(X:B)_\rho - 
    \max_{i \in [t]} \{I^\epsilon_\infty(X:E_i)_\rho\} - 
    \frac{4 \log t}{\epsilon^2}.
\]
Then there exists a private classical code 
that can send classical messages $m \in [2^R]$ over the wiretap channel 
$\cT$ such that $B$ can recover each message $m$ with error probability 
at most $2 \epsilon$ ({\em correctness}) and for all $i \in [t]$, 
the state $\sigma^{E_i}(m)$ of eavesdropper $E_i$ satisfies, for each
$m$, $\|\sigma^{E_i}(m) - \rho^{E_i}\| < 4 \epsilon$ ({\em privacy}).
\end{theorem}

Similarly, Theorem~\ref{thm:smoothcoveringconc} can be used to prove 
good eavesdropper dimension independent inner bounds for multiterminal 
wiretap channels with many non-interacting eavesdroppers. For example,
for the wiretap QMAC we prove 
Theorem~\ref{thm:wiretapQMACconc} below, which is an extension 
of \cite[Theorem~5]{Sen:telescoping} that had only one eavesdropper with
a very similar proof.
\begin{theorem}
\label{thm:wiretapQMACconc}
Let $\cN^{A B \rightarrow C E_1 \cdots E_t}$ denote a wiretap QMAC from two
senders Alice, Bob to a single legitimate receiver Charlie and $t$ 
non-interacting eavesdroppers $E_1 \cdots E_t$. Alice, Bob would like 
to send classical
messages $m \in 2^{R_1}$, $n \in 2^{R_2}$ respectively to Charlie
by using the channel $\cN$ in such a way that each $E_i$ gets almost
no information about $(m,n)$. 
Let $X$, $Y$ be new classical
alphabets. Let $Q$ be a new `timesharing' alphabet. Put a normalised joint 
probability distribution
on $Q \times X \times Y$ of the form $p(q) p(x|q) p(y|q)$ i.e.
the distributions on $X$ and $Y$ are independent conditioned on 
any $q \in Q$. Fix classical to quantum encodings $x \mapsto \alpha_x^A$,
$y \mapsto \beta^B_y$.  
Define the classical quantum `control state':
\[
\sigma^{Q X Y C E_1 \cdots E_t} :=
\sum_{q \in Q} \sum_{x \in X} \sum_{y \in Y}
p(q) p(x|q) p(y|q) \ketbra{q,x,y}^{QXY} \otimes 
\cN^{A B\rightarrow C E_1 \cdots E_t} (\alpha^A_x \otimes \beta^B_y).
\]
Let the rates $R_1$, $R_2$ satisfy the following inequalities.
\begin{eqnarray*}
R_1
& < & 
I^\epsilon_H(X : YC | Q)_\sigma - 
\max_{i \in [t]}\{I^\epsilon_\infty(X : E_i | Q)_\sigma\} - 
\frac{4 \log t}{\epsilon^2}, \\
R_2
& < & 
I^\epsilon_H(Y : XC | Q)_\sigma - 
\max_{i \in [t]}\{I^\epsilon_\infty(Y : E_i | Q)_\sigma\} - 
\frac{4 \log t}{\epsilon^2}, \\
R_1 + R_2
& < & 
I^\epsilon_H(XY : C | Q)_\sigma - 
\max_{i \in [t]}\{I^\epsilon_\infty(X Y : E_i | Q)_\sigma\} - 
\frac{4 \log t}{\epsilon^2}.
\end{eqnarray*}
Then,
\begin{eqnarray*}
\E_{m,n}[\mbox{probability Charlie decodes $(m,n)$ incorrectly}] 
& < &
50 \sqrt{\epsilon}
~~~~~ \cdots ~
\mbox{accurate transmission}, \\
\E_{m,n}[\max_{i \in [t]}\{\|\sigma^{E_i}_{m,n} - \sigma^{E_i}\}\|_1] 
& < & 
16 \sqrt{\epsilon}
~~~~~ \cdots ~
\mbox{high privacy},
\end{eqnarray*}
where $\E_{m,n}[\cdot]$ denotes the expectation over a uniform
choice of message pair $(m,n) \in [2^{R_1}] \times [2^{R_2}]$,
$\sigma^{E_i}_{m,n}$ denotes $E_i$'s state when the message pair 
$(m,n)$ is sent, and $\sigma^{E_i}$ denotes the marginal of the control
state on $E_i$.
\end{theorem}
The asymptotic iid limit of the above theorem is is now immediate.
\begin{corollary}
\label{cor:wiretapQMACconc}
In the asymtotic iid limit of a wiretap QMAC, the rate pairs 
per channel use satisfying the following inequalities are achievable.
\begin{eqnarray*}
R_1
& < & 
I(X : YC | Q)_\sigma - 
\max_{i \in [t]}\{I(X : E_i | Q)_\sigma\}, \\
R_2
& < & 
I(Y : XC | Q)_\sigma - 
\max_{i \in [t]}\{I(Y : E_i | Q)_\sigma\}, \\
R_1 + R_2
& < & 
I(XY : C | Q)_\sigma - 
\max_{i \in [t]}\{I(XY : E_i | Q)_\sigma\}.
\end{eqnarray*}
\end{corollary}

\section{Smooth expander matrix Chernoff bound}
In this section, we prove our new smooth expander matrix Chernoff bound
or in other words, our new smooth unipartite classical quantum
soft covering lemma in concentration when the samples are taken from
an expander walk. But first, we have to prove a new  smooth unipartite 
classical quantum
soft covering lemma in expectation when the samples are taken from
an expander walk. 
\begin{theorem}[{\bf Smooth unipartite expander soft covering
in expectation}]
\label{thm:expandersoftcoveringexpectaton}
Let $X$ be a classical alphabet and $M$ a Hilbert space. Let $G$ be a
constant degree expander graph with vertex set $X$. Let the second largest
eigenvalue in absolute value of $G$ have absolute value $\lambda < 1/4$.
Let $p^X$ be
a normalised probability distribution on $X$ and $x \mapsto \rho_x^M$ be
a classical to quantum mapping where $\rho_x^M$ is a normalised quantum
state. Define the control state 
$
\rho^{XM} := \sum_x p(x) \ketbra{x}^X \otimes \rho_x^M.
$
Let $\epsilon$ be positive and sufficiently small. Then,
\[
\E_{x_1, \ldots, x_K}
\left[
\|\frac{|X|}{K} \sum_{i=1}^K p(x_i) \rho_{x_i}^M - \rho^M\|_1
\right]
< 2 \sqrt{\epsilon},
\]
where the expecation is taken over a stationary random walk
$x_1, \ldots, x_K$ on $G$, if
\[
\log K > \log |X| + \log \log |X| - \Hmin^\epsilon(X | M)_\rho
+ \log \epsilon^{-1}.
\]
\end{theorem}
\begin{proof}
First, smooth $p^X$ to the subnormalised probability distribution
$p^{'X}$ and smooth $\rho_x^M$ to the normalised quantum state 
$\rho_x^{'M}$ that achieves the minimum in the definition of
$D^\epsilon_\infty(\rho^{XM} \| \frac{1^X}{|X|} \otimes \rho^M)$.
Define the subnormalised classical quantum state
$\rho^{'X M} := \sum_x p'(x) \ketbra{x}^X \otimes \rho_x^{'M}$.
By the discussion in Section~\ref{sec:prelims}, 
\begin{equation}
\label{eq:expanderopineq}
(\forall x \in X: 
p'(x) \rho_x^{'M} \leq 
2^{-\Hmin^\epsilon(X|M)_\rho} \rho^{'M})
~~ \mbox{AND} ~~
\|\rho^{'XM} - \rho^{XM}\|_1 \leq \epsilon.
\end{equation}
It suffices to show
\[
\E_{x_1, \ldots, x_K}
\left[
\|\frac{|X|}{K} \sum_{i=1}^K p'(x_i) \rho_{x_i}^{'M} - \rho^{'M}\|_1
\right]
< 2 \sqrt{\epsilon} - 2\epsilon,
\]
because
\begin{eqnarray*}
\lefteqn{
\E_{x_1, \ldots, x_K}
\left[
\|\frac{|X|}{K} \sum_{i=1}^K p(x_i) \rho_{x_i}^M - \rho^M\|_1
\right]
} \\
& \leq &
\E_{x_1, \ldots, x_K}
\left[
\|\frac{|X|}{K} \sum_{i=1}^K p'(x_i) \rho_{x_i}^{'M} - \rho^{'M}\|_1
\right] +
\|\rho^{'M} - \rho^M\|_1 \\
&  &
{} +
\E_{x_1, \ldots, x_K}
\left[
\|\frac{|X|}{K} \sum_{i=1}^K p'(x_i) \rho_{x_i}^{'M} - 
  \frac{|X|}{K} \sum_{i=1}^K p(x_i) \rho_{x_i}^{M} \|_1 
\right] \\
& \leq &
\E_{x_1, \ldots, x_K}
\left[
\|\frac{|X|}{K} \sum_{i=1}^K p'(x_i) \rho_{x_i}^{'M} - \rho^{'M}\|_1
\right] +
\|\rho^{'M} - \rho^M\|_1 \\
&  &
{} +
\frac{|X|}{K} \sum_{i=1}^K
\E_{x_1, \ldots, x_K}
\left[
\|p'(x_i) \rho_{x_i}^{'M} - p(x_i) \rho_{x_i}^{M} \|_1 
\right] \\
&   =  &
\E_{x_1, \ldots, x_K}
\left[
\|\frac{|X|}{K} \sum_{i=1}^K p'(x_i) \rho_{x_i}^{'M} - \rho^{'M}\|_1
\right] +
\|\rho^{'M} - \rho^M\|_1 +
\frac{|X|}{K} \sum_{i=1}^K
\E_{x_i}
\left[
\|p'(x_i) \rho_{x_i}^{'M} - p(x_i) \rho_{x_i}^{M} \|_1 
\right] \\
&   =  &
\E_{x_1, \ldots, x_K}
\left[
\|\frac{|X|}{K} \sum_{i=1}^K p'(x_i) \rho_{x_i}^{'M} - \rho^{'M}\|_1
\right] +
\|\rho^{'M} - \rho^M\|_1 \\
&  &
{} +
\frac{|X|}{K} \cdot |K| 
\sum_{x \in X} \frac{1}{|X|}
\|p'(x) \rho_{x}^{'M} - p(x) \rho_{x}^{M} \|_1 \\
&   =  &
\E_{x_1, \ldots, x_K}
\left[
\|\frac{|X|}{K} \sum_{i=1}^K p'(x_i) \rho_{x_i}^{'M} - \rho^{'M}\|_1
\right] +
\|\rho^{'M} - \rho^M\|_1 + \|\rho^{'XM} - \rho^{XM}\|_1 \\
& \leq &
\E_{x_1, \ldots, x_K}
\left[
\|\frac{|X|}{K} \sum_{i=1}^K p'(x_i) \rho_{x_i}^{'M} - \rho^{'M}\|_1
\right] +
2 \|\rho^{'XM} - \rho^{XM}\|_1 \\
& \leq &
\E_{x_1, \ldots, x_K}
\left[
\|\frac{|X|}{K} \sum_{i=1}^K p'(x_i) \rho_{x_i}^{'M} - \rho^{'M}\|_1
\right] +
2 \epsilon.
\end{eqnarray*}
In the second equality above we used the fact that, for all $i \in [K]$,
the distribution of $x_i$ in a stationary random walk is uniform.

By Fact~\ref{fact:matrixCauchySchwarz}, it suffices to show 
\[
\E_{x_1, \ldots, x_K}
\left[
\|
\frac{|X|}{K} \sum_{i=1}^K p'(x_i) 
(\rho^{'M})^{-1/4} \rho_{x_i}^{'M} (\rho^{'M})^{-1/4}
- (\rho^{'M})^{1/2}
\|_2
\right]
< 2 \sqrt{\epsilon} - 2\epsilon.
\]
By convexity of the squaring function, it suffices to show
\[
\E_{x_1, \ldots, x_K}
\left[
\|
\frac{|X|}{K} \sum_{i=1}^K p'(x_i) 
(\rho^{'M})^{-1/4} \rho_{x_i}^{'M} (\rho^{'M})^{-1/4}
- (\rho^{'M})^{1/2}
\|_2^2
\right]
< (2 \sqrt{\epsilon} - 2\epsilon)^2.
\]
The left hand side of the above inequality satisfies
\begin{eqnarray*}
\lefteqn{
\E_{x_1, \ldots, x_K}
\left[
\|
\frac{|X|}{K} \sum_{i=1}^K p'(x_i) 
(\rho^{'M})^{-1/4} \rho_{x_i}^{'M} (\rho^{'M})^{-1/4}
- (\rho^{'M})^{1/2}
\|_2^2
\right]
} \\
& = &
\frac{|X|^2}{K^2} 
\E_{x_1, \ldots, x_K}
\left[
\|
\sum_{i=1}^K p'(x_i) 
(\rho^{'M})^{-1/4} \rho_{x_i}^{'M} (\rho^{'M})^{-1/4}
\|_2^2
\right] \\
&  &
{} - 
\frac{2 |X|}{K} 
\E_{x_1, \ldots, x_K}
\left[
\Tr
\left[
\left(
\sum_{i=1}^K p'(x_i) 
(\rho^{'M})^{-1/4} \rho_{x_i}^{'M} (\rho^{'M})^{-1/4}
\right)
(\rho^{'M})^{1/2}
\right]
\right] +
\|(\rho^{'M})^{1/2}\|_2^2 \\
& \leq &
\frac{|X|^2}{K^2} 
\E_{x_1, \ldots, x_K}
\left[
\Tr
\left[
\left(
\sum_{i=1}^K p'(x_i) 
(\rho^{'M})^{-1/4} \rho_{x_i}^{'M} (\rho^{'M})^{-1/4}
\right)
\left(
\sum_{j=1}^K p'(x_j) 
(\rho^{'M})^{-1/4} \rho_{x_j}^{'M} (\rho^{'M})^{-1/4}
\right)
\right]
\right] \\
&  &
{} - 
\frac{2 |X|}{K} 
\E_{x_1, \ldots, x_K}
\left[
\Tr
\left[
\sum_{i=1}^K p'(x_i) \rho_{x_i}^{'M}
\right]
\right] +
\Tr [\rho^{'M}] \\
&   =  &
\frac{|X|^2}{K^2} 
\sum_{i=1}^K \sum_{j=1}^K 
\E_{x_1, \ldots, x_K}
\left[
p'(x_i) p'(x_j) 
\Tr
\left[
\left(
(\rho^{'M})^{-1/4} \rho_{x_i}^{'M} (\rho^{'M})^{-1/4}
\right)
\left(
(\rho^{'M})^{-1/4} \rho_{x_j}^{'M} (\rho^{'M})^{-1/4}
\right)
\right]
\right] \\
&  &
{} - 
\frac{2 |X|}{K} \sum_{i=1}^K \E_{x_1, \ldots, x_K} [p'(x_i)] +
\Tr [\rho^{'M}] \\
&   =  &
\frac{|X|^2}{K^2} 
\sum_{i=1}^K \sum_{j=1}^K 
\E_{x_i, x_j}
\left[
p'(x_i) p'(x_j) 
\Tr
\left[
\left(
(\rho^{'M})^{-1/4} \rho_{x_i}^{'M} (\rho^{'M})^{-1/4}
\right)
\left(
(\rho^{'M})^{-1/4} \rho_{x_j}^{'M} (\rho^{'M})^{-1/4}
\right)
\right]
\right] \\
&  &
{} - 
\frac{2 |X|}{K} \sum_{i=1}^K \E_{x_i} \frac{p'(x_i)}{|X|} +
\Tr [\rho^{'M}] \\
&   =  &
\frac{|X|^2}{K^2} 
\sum_{i=1}^K \sum_{j=1}^K 
\E_{x_i, x_j}
\left[
p'(x_i) p'(x_j) 
\Tr
\left[
\left(
(\rho^{'M})^{-1/4} \rho_{x_i}^{'M} (\rho^{'M})^{-1/4}
\right)
\left(
(\rho^{'M})^{-1/4} \rho_{x_j}^{'M} (\rho^{'M})^{-1/4}
\right)
\right]
\right] \\
&  &
{} - 
\frac{2 |X|}{K} \cdot |K| \sum_{x \in X} \frac{p'(x)}{|X|} +
\Tr [\rho^{'M}] \\
&   =  &
\frac{|X|^2}{K^2} 
\sum_{i=1}^K \sum_{j=1}^K 
\E_{x_i, x_j}
\left[
p'(x_i) p'(x_j) 
\Tr
\left[
\left(
(\rho^{'M})^{-1/4} \rho_{x_i}^{'M} (\rho^{'M})^{-1/4}
\right)
\left(
(\rho^{'M})^{-1/4} \rho_{x_j}^{'M} (\rho^{'M})^{-1/4}
\right)
\right]
\right] \\
&  &
{} - 
2 \Tr [\rho^{'XM}] + \Tr [\rho^{'XM}] \\
& \leq &
\frac{|X|^2}{K^2} 
\sum_{i=1}^K \sum_{j=1}^K 
\E_{x_i, x_j}
\left[
p'(x_i) p'(x_j) 
\Tr
\left[
\left(
(\rho^{'M})^{-1/4} \rho_{x_i}^{'M} (\rho^{'M})^{-1/4}
\right)
\left(
(\rho^{'M})^{-1/4} \rho_{x_j}^{'M} (\rho^{'M})^{-1/4}
\right)
\right]
\right] \\
&  &
{} - 
\Tr [\rho^{XM}] + \|\rho^{'XM} - \rho^{XM}\|_1 \\
& \leq &
\frac{|X|^2}{K^2} 
\sum_{i=1}^K \sum_{j=1}^K 
\E_{x_i, x_j}
\left[
p'(x_i) p'(x_j) 
\Tr
\left[
\left(
(\rho^{'M})^{-1/4} \rho_{x_i}^{'M} (\rho^{'M})^{-1/4}
\right)
\left(
(\rho^{'M})^{-1/4} \rho_{x_j}^{'M} (\rho^{'M})^{-1/4}
\right)
\right]
\right] \\
&  &
-  1 + \epsilon.
\end{eqnarray*}
In the fourth equality above we used the fact that, for all
$i \in [K]$, the distribution of $x_i$ in a stationary random walk
is uniform.

Hence it suffices to show
\begin{equation}
\label{eq:FirstTermInLHS}
\begin{array}{rcl}
\lefteqn{
\frac{|X|^2}{K^2} 
\sum_{i=1}^K \sum_{j=1}^K 
\E_{x_i, x_j}
\left[
p'(x_i) p'(x_j) 
\Tr
\left[
\left(
(\rho^{'M})^{-1/4} \rho_{x_i}^{'M} (\rho^{'M})^{-1/4}
\right)
\left(
(\rho^{'M})^{-1/4} \rho_{x_j}^{'M} (\rho^{'M})^{-1/4}
\right)
\right]
\right] 
} \\
& \leq &
1 + 3\epsilon - 8\epsilon^{3/2} + 4\epsilon^2.
~~~~~~~~~~~~~~~~~~~~~~~~~~~~~~~~~~~~~~~~~~~~~~~~~~~~~~~~~~~~~~~~~~~~~~~~~~
\end{array}
\end{equation}
Consider a term like
\[
\E_{x_i, x_j}
\left[
p'(x_i) p'(x_j) 
\Tr
\left[
\rho_{x_i}^{'M} (\rho^{'M})^{-1/2}
\rho_{x_j}^{'M} (\rho^{'M})^{-1/2}
\right]
\right]. 
\]
To handle it, we consider two cases as follows.

The first case is when $|i - j| \leq \frac{\log |X|}{\log \lambda^{-1}}$.
The number of such terms is at most
$\frac{K \log |X|}{\log \lambda^{-1}}$.
\begin{eqnarray*}
\lefteqn{
\E_{x_i, x_j}
\left[
p'(x_i) p'(x_j) 
\Tr
\left[
\rho_{x_i}^{'M} (\rho^{'M})^{-1/2}
\rho_{x_j}^{'M} (\rho^{'M})^{-1/2}
\right]
\right]
} \\
& \leq &
2^{-\Hmin^\epsilon(X|M)_\rho}
\E_{x_i, x_j}
\left[
p'(x_i) 
\Tr
\left[
\rho_{x_i}^{'M} (\rho^{'M})^{-1/2}
\rho^{'M} (\rho^{'M})^{-1/2}
\right]
\right] 
    =   
2^{-\Hmin^\epsilon(X|M)_\rho} \E_{x_i} [p'(x_i)] \\
&  =   &
2^{-\Hmin^\epsilon(X|M)_\rho}  
\sum_{x \in X} \frac{p'(x)}{|X|} 
 \leq 
\frac{2^{-\Hmin^\epsilon(X|M)_\rho}}{|X|}.
\end{eqnarray*}
Above, we made use of Equation~\ref{eq:expanderopineq} in the first
inequality, and in the second equality we used the fact that, 
for any $i \in [K]$,  the distribution of
$x_i$ in a stationary random walk is uniform.

The second case is when $|i - j| > \frac{\log |X|}{\log \lambda^{-1}}$.
Define 
$t := \lceil |i - j| - \frac{\log |X|}{\log \lambda^{-1}} \rceil$. 
Then $t$ is an integer satisfying
$1 \leq t \leq K - \frac{\log |X|}{ \log \lambda^{-1}}$.
For a given $t$, the number of such terms is at most
$2(K - t - \frac{\log |X|}{\log \lambda^{-1}})$. We analyse the case
for a given $t$ as follows. Fix a value $x'$ for $x_i$. Let $q^X$ denote
the probability distribution of $x'_j$ given $x_i = x'$. By reversibility
of the expander walk, it does not matter whether $i < j$ or $i > j$.
So in the analysis below, we will tacitly assume that $i < j$. By
Fact~\ref{fact:expanderwalk}, $q^X = \frac{1^X}{|X|} + q^{'X}$
where $\braket{1^X}{q^{'X}} = 0$ and 
\[
\|q^{'X}\|_1 \leq 
|X|^{1/2} \|q^{'X}\|_2 \leq 
|X|^{1/2} \lambda^{t + \frac{\log |X|}{\log \lambda^{-1}} - 1} \leq
|X|^{1/2} \lambda^{t + \frac{\log |X|}{2 \log \lambda^{-1}}} \leq
\lambda^{t}. 
\]
So,
\begin{eqnarray*}
\lefteqn{
\E_{x_j | x_i = x'}
\left[
p'(x_i) p'(x_j) 
\Tr
\left[
\rho_{x_i}^{'M} (\rho^{'M})^{-1/2}
\rho_{x_j}^{'M} (\rho^{'M})^{-1/2}
\right]
\right]
} \\
&   =  &
p'(x') 
\sum_{x \in X} q(x) p'(x) 
\Tr
\left[
\rho_{x'}^{'M} (\rho^{'M})^{-1/2}
\rho_{x}^{'M} (\rho^{'M})^{-1/2}
\right] \\
&   =  &
\frac{p'(x')}{|X|}
\sum_{x \in X} p'(x) 
\Tr
\left[
\rho_{x'}^{'M} (\rho^{'M})^{-1/2}
\rho_{x}^{'M} (\rho^{'M})^{-1/2}
\right] \\
&  &
{} +
p'(x')
\sum_{x \in X} q'(x) p'(x) 
\Tr
\left[
\rho_{x'}^{'M} (\rho^{'M})^{-1/2}
\rho_{x}^{'M} (\rho^{'M})^{-1/2}
\right] \\
&  =   &
\frac{p'(x')}{|X|}
\Tr
\left[
\rho_{x'}^{'M} (\rho^{'M})^{-1/2}
\rho^{'M} (\rho^{'M})^{-1/2}
\right] \\
&  &
{} +
p'(x')
\sum_{x \in X} q'(x) p'(x) 
\Tr
\left[
\rho_{x'}^{'M} (\rho^{'M})^{-1/2}
\rho_{x}^{'M} (\rho^{'M})^{-1/2}
\right] \\
& \leq &
\frac{p'(x')}{|X|} +
p'(x') 2^{-\Hmin^\epsilon(X|M)_\rho}
\sum_{x \in X} |q'(x)|
\Tr
\left[
\rho_{x'}^{'M} (\rho^{'M})^{-1/2}
\rho^{'M} (\rho^{'M})^{-1/2}
\right] \\
&   =  &
\frac{p'(x')}{|X|} +
p'(x') 2^{-\Hmin^\epsilon(X|M)_\rho} \|q^{'X}\|_1 
 \leq 
\frac{p'(x')}{|X|} +
p'(x') 2^{-\Hmin^\epsilon(X|M)_\rho} \lambda^t,
\end{eqnarray*}
where we used Equation~\ref{eq:expanderopineq} in the first inequality.
Hence, 
\begin{eqnarray*}
\lefteqn{
\E_{x_i, x_j}
\left[
p'(x_i) p'(x_j) 
\Tr
\left[
\rho_{x_i}^{'M} (\rho^{'M})^{-1/2}
\rho_{x_j}^{'M} (\rho^{'M})^{-1/2}
\right]
\right]
} \\
& \leq &
\E_{x_i}
\left[
\frac{p'(x_i)}{|X|} +
p'(x_i) 2^{-\Hmin^\epsilon(X|M)_\rho} \lambda^t
\right] 
   =  
\sum_{x \in X} \frac{p'(x)}{|X|}
\left(\frac{1}{|X|} + 2^{-\Hmin^\epsilon(X|M)_\rho} \lambda^t\right) \\
& \leq &
\frac{1}{|X|}
\left(\frac{1}{|X|} + 2^{-\Hmin^\epsilon(X|M)_\rho} \lambda^t\right),
\end{eqnarray*}
where in the equality above we used the fact that, for any $i \in [K]$,  
the distribution of
$x_i$ in a stationary random walk is uniform.

We can now upper bound the left hand side of 
Equation~\ref{eq:FirstTermInLHS} as follows:
\begin{eqnarray*}
\lefteqn{
\frac{|X|^2}{K^2} 
\sum_{i=1}^K \sum_{j=1}^K 
\E_{x_i, x_j}
\left[
p'(x_i) p'(x_j) 
\Tr
\left[
\left(
(\rho^{'M})^{-1/4} \rho_{x_i}^{'M} (\rho^{'M})^{-1/4}
\right)
\left(
(\rho^{'M})^{-1/4} \rho_{x_j}^{'M} (\rho^{'M})^{-1/4}
\right)
\right]
\right] 
} \\
& \leq &
\frac{|X|^2}{K^2} 
\left(
\frac{K \log |X|}{\log \lambda^{-1}} \cdot
\frac{2^{-\Hmin^\epsilon(X|M)_\rho}}{|X|} 
\right. \\
& &
\left.
{} +
\sum_{t = 1}^{K - \frac{\log |X|}{\log \lambda^{-1}}}
2\left(K - t - \frac{\log |X|}{\log \lambda^{-1}}\right) 
\frac{1}{|X|}
\left(\frac{1}{|X|} + 2^{-\Hmin^\epsilon(X|M)_\rho} \lambda^t\right)
\right) \\
& <    &
\frac{|X|^2}{K^2} 
\left(
\frac{K \log |X|}{\log \lambda^{-1}} \cdot
\frac{2^{-\Hmin^\epsilon(X|M)_\rho}}{|X|} +
\sum_{t = 1}^{K - 1}
\frac{2 (K - t)}{|X|}
\left(\frac{1}{|X|} + 2^{-\Hmin^\epsilon(X|M)_\rho} \lambda^t\right)
\right) \\
& <    &
\frac{|X|^2}{K^2} 
\left(
\frac{K \log |X|}{\log \lambda^{-1}} \cdot
\frac{2^{-\Hmin^\epsilon(X|M)_\rho}}{|X|} +
\sum_{t = 1}^{K - 1}
\frac{2 (K - t)}{|X|^2} +
\frac{2 K}{|X|} \cdot 2^{-\Hmin^\epsilon(X|M)_\rho} \cdot 
\frac{\lambda}{1 - \lambda}
\right) \\
& <    &
\frac{|X|^2}{K^2} 
\left(
\frac{2 K \log |X|}{\log \lambda^{-1}} \cdot
\frac{2^{-\Hmin^\epsilon(X|M)_\rho}}{|X|} +
\frac{2}{|X|^2} \cdot \frac{K (K-1)}{2} 
\right) \\ 
& <    &
\frac{|X|}{K} \cdot
\frac{2 \log |X|}{\log \lambda^{-1}} \cdot
\frac{2^{-\Hmin^\epsilon(X|M)_\rho}}{|X|} + 1 
 =    
\frac{2}{\log \lambda^{-1}} \cdot
\frac{2^{\log |X| + \log \log |X| - \Hmin^\epsilon(X|M)_\rho}}{K} + 1 \\
& <    &
\frac{2^{\log |X| + \log \log |X| - \Hmin^\epsilon(X|M)_\rho}}{K} + 1 
   <
1 + \epsilon 
   <
1 + 3 \epsilon - 8 \epsilon^{3/2} + 4 \epsilon^2.
\end{eqnarray*}
Above, we used the lower bound on $\log K$ assumed in the statement
of the theorem and small enough $\epsilon$.

This completes the proof of the theorem.
\end{proof}

Next, we need to define functions satisfying the 
{\em bounded excision condition}.
\begin{definition}[{\bf (Bounded excision)}]
\label{def:boundedexcision}
Let $K$ be a positive integer. Suppose there is a family of functions
$f_i: X^i \rightarrow \R$ for $i \in [K]$. This family is said
to satisfy {\em bounded excision with parameters $c$, $c_{l_1 l_2}$ for
$1 \leq l_1 \leq l_2 \leq K$} if for all pairs $(l_1, l_2)$, there
exist functions $g_{1, l_1, l_2}: X^{l_2 - l_1 + 1} \rightarrow \R$,
$g_{2, l_1, l_2}: X^{l_2 - l_1 + 1} \rightarrow \R$ such that,
for all $(x_1, \ldots, x_K) \in X^K$,
\begin{eqnarray*}
g_{2, l_1, l_2}(x_{l_1}, \ldots, x_{l_2}) 
& \leq &
f_K(x_1, \ldots, x_{l_1}, \ldots, x_{l_2}, \ldots, x_K) \\
&  &
{} -
f_{K-l_2+l_1-1}(x_1, \ldots, x_{l_1 - 1}, x_{l_2 + 1}, \ldots, x_{K}) \\
& \leq & 
g_{1, l_1, l_2}(x_{l_1}, \ldots, x_{l_2}), \\
g_{1, l_1, l_2}(x_{l_1}, \ldots, x_{l_2}) 
& \leq & c_{1, l_1 l_2}, ~
g_{2, l_1, l_2}(x_{l_1}, \ldots, x_{l_2}) \geq c_{2, l_1 l_2}, ~
c_{l_1 l_2} := c_{1, l_1 l_2} - c_{2, l_1 l_2}, \\
|f_{K-l_2+l_1-1}(x_1, \ldots, x_{l_1 - 1}, x_{l_2 + 1}, \ldots, x_{K})| 
& \leq & 
c.
\end{eqnarray*}
\end{definition}

We now prove our concentration result under expander walks for function 
families satisfying bounded excision. It can be viewed as a generalisation
of McDiarmid's method of bounded differences, which requires independent
sampling, to sampling via an expander walk.
\begin{theorem}
\label{thm:boundedexcisionconc}
Let $G$ be a
constant degree expander graph with vertex set $X$. Let the second largest
eigenvalue in absolute value of $G$ have absolute value $\lambda$.
Let $K$ be a positive integer. Suppose there is a function family
$f_i: X^i \rightarrow \R$, $1 \leq i \leq K$, satisfying bounded
excision with parameters $c$, $c_{l_1 l_2}$ for 
$1 \leq l_1 \leq l_2 \leq K$. Let $\epsilon > 0$. Then,
\[
\Pr_{x_1, \ldots, x_K}[
|f_K(x_1, \ldots, x_K) - \E_{z_1, \ldots, z_K}[f_K(z_1, \ldots, z_K)]|
\geq \epsilon
] \leq
2 \exp\left(-\frac{2\epsilon^2}{\sum_{i=1}^K d_i^2}\right),
\]
where
$
d_i := 2 c_{i,a+i} + c_{a+i+1, a+i+b},
$
$
a := \lceil \frac{\log |X|}{\log \lambda^{-1}} \rceil,
$
$
b := \lceil \frac{\log (c / c_{i,a+i})}{\log \lambda^{-1}} \rceil,
$
and the probability and expectation above are taken via a stationary
random walk of length $K$ on $G$.
\end{theorem}
\begin{proof}
We follow the general outline of Ying's method \cite{Ying:bounded} giving
an elementary proof of Fact~\ref{fact:mcdiarmid}. However, Ying's method 
required independent samples and so we have to suitably modify our
strategy in order to handle sampling via an expander walk. We will
show a concentration upper bound for the upper tail. Concentration
upper bound for the lower tail can be proved similarly. Combining
the two concentration bounds gives the claim of the theorem.

For $1 \leq i \leq K$, define a function $h_i: X^i \rightarrow \R$ by
\[
h_i(x_1, \ldots, x_i) :=
\E_{z_{i+1}, \ldots, z_K}[f_K(x_1, \ldots, x_i, z_{i+1}, \ldots, z_K)] -
\E_{z_{i}, \ldots, z_K}[f_K(x_1, \ldots, x_{i-1}, z_{i}, \ldots, z_K)],
\]
where the expectations are taken over random walks on $G$
starting from $x_i$ and $x_{i-1}$ respectively.
Observe that for any $(x_1, \ldots, x_K) \in X^K$,
\[
\sum_{i=1}^K h_i(x_1, \ldots, x_i) =
f_K(x_1, \ldots, x_K) -
\E_{z_1, \ldots, z_K}[f_K(z_1, \ldots, z_K)] 
~ \mbox{AND} ~
\E_{x_i}[h_i(x_1, \ldots, x_i)] = 0,
\]
where the expectations are taken over a stationary random walk of 
length $K$ on $G$ and a random choice of a neighbour of $x_{i-1}$
respectively. Let $\theta > 0$. Then,
\begin{eqnarray*}
\lefteqn{
\Pr_{x_1, \ldots, x_K}[
f_K(x_1, \ldots, x_K) - \E_{z_1, \ldots, z_K}[f_K(z_1, \ldots, z_K)]
\geq \epsilon
] 
} \\
& = &
\Pr_{x_1, \ldots, x_K}\left[
\sum_{i=1}^K h_i(x_1, \ldots, x_i) \geq \epsilon
\right] 
 \leq 
\Pr_{x_1, \ldots, x_K}\left[
\exp\left(\theta \sum_{i=1}^K h_i(x_1, \ldots, x_i)\right) 
\geq e^{\theta \epsilon}
\right] \\ 
& \leq &
e^{-\theta \epsilon}
\E_{x_1, \ldots, x_K}
\left[\exp\left(\theta \sum_{i=1}^K h_i(x_1, \ldots, x_i)\right)\right].
\end{eqnarray*}
We have,
\begin{eqnarray*}
\lefteqn{
\E_{x_1, \ldots, x_K}
\left[\exp\left(\theta \sum_{i=1}^K h_i(x_1, \ldots, x_i)\right)\right]
} \\
& = &
\E_{x_1, \ldots, x_{K-1}}
\left[
\exp\left(\theta \sum_{i=1}^{K-1} h_i(x_1, \ldots, x_i)\right)
\E_{x_K}[
e^{\theta h_K(x_1, \ldots, x_{K-1}, x_K)}
\right],
\end{eqnarray*}
where the second expectation in the right hand size of the equality
is taken over a random choice of a neighbour $x_K$ of vertex $x_{K-1}$.
Now for any fixed values for $x_1, \ldots, x_{K-1}$, the random variable
$Y := h_K(x_1, \ldots, x_{K-1}, x_K)$, where the randomness comes from
the choice of $x_K$ given a fixed $x_{K-1}$, satisfies the
conditions of Fact~\ref{fact:HoeffdingLemma} with 
$b - a \leq c_{K,K} \leq d_K$. Hence,
\begin{eqnarray*}
\lefteqn{
\E_{x_1, \ldots, x_K}
\left[\exp\left(\theta \sum_{i=1}^K h_i(x_1, \ldots, x_i)\right)\right]
} \\
& \leq &
\E_{x_1, \ldots, x_{K-1}}
\left[
\exp\left(\theta \sum_{i=1}^{K-1} h_i(x_1, \ldots, x_i)\right)
\right] \cdot
\exp\left(\frac{\theta^2 d_K^2}{8}\right).
\end{eqnarray*}
We will argue similarly for each $i = K-1, \ldots, 1$ to finally show
\[
\E_{x_1, \ldots, x_K}
\left[\exp\left(\theta \sum_{i=1}^K h_i(x_1, \ldots, x_i)\right)\right]
\leq
\exp\left(\frac{\theta^2 \sum_{i=1}^K d_i^2}{8}\right).
\]
We only have to show that at stage $i$,
for all $(x_1, \ldots, x_{i-1}) \in X^{i-1}$,
\[
\max_{x_i \in X} h_i(x_1, \ldots, x_i) -
\min_{x_i \in X} h_i(x_1, \ldots, x_i) 
\leq
d_i.
\]

Fix values for $x_1, \ldots, x_{i-1}$. 
Suppose the maximum for $h_i(x_1, \ldots, x_{i-1}, x_i)$ is attained
at $x_i(1)$ and the minimum at $x_i(2)$. The issue here is that the
expander walk $z_{i+1}, \ldots, z_K$ on $G$ starting at 
$z_i = x_i(1)$ has a different probability distribution than the
walk starting at $z_i = x_i(2)$. This is where we use bounded excision
property (Definition~\ref{def:boundedexcision}) and the rapidly
mixing property of expander walks (Fact~\ref{fact:expanderwalk}) to
obtain the upper bound of $d_i$ at stage $i$ as desired above.
Let $p^X(1, t)$ be the probability distribution of vertex $z_{i+t}$
of the random walk on $G$ starting from $z_i = x_i(1)$. Similarly,
we define the probability distribution $p^X(2, t)$. By 
Fact~\ref{fact:expanderwalk}, we have
\[
p^X(1,t) = \frac{1^X}{|X|} + v^X(1,t), ~~
p^X(2,t) = \frac{1^X}{|X|} + v^X(2,t),
\]
where $v^X(1,t)$, $v^X(2,t)$ are vectors in $\R^X$ satisfying
$\|v^X(1,t)\|_2 \leq \lambda^t$, $\|v^X(2,t)\|_2 \leq \lambda^t$. So
\[
\|p^X(1,t) - p^X(2,t)\|_2 \leq 2\lambda^t \implies
\|p^X(1,t) - p^X(2,t)\|_1 \leq 2\lambda^t |X|^{1/2}.
\]
Hence at $t = a = \lceil \frac{\log |X|}{\log \lambda^{-1}} \rceil$,
$\|p^X(1,t) - p^X(2,t)\|_1 \leq 1$. Thus at $t = a + l$,
$\|p^X(1,t) - p^X(2,t)\|_1 \leq \lambda^l$. 
For $l = b = \lceil \frac{\log (c / c_{i,a+i})}{\log \lambda^{-1}} \rceil$,
$\|p^X(1,t) - p^X(2,t)\|_1 \leq \frac{c_{i,i+a}}{c}$. 

By Definition~\ref{def:boundedexcision},
\begin{eqnarray*}
\lefteqn{
h_i(x_1, \ldots, x_{i-1}, x_i(1)) - h_i(x_1, \ldots, x_{i-1}, x_i(2))
} \\
& = &
\E_{\substack{z_{i+1}, \ldots, z_K \\ \mathrm{from}~x_i(1)}}[
f_K(x_1, \ldots, x_{i-1}, x_i(1), z_{i+1}, \ldots, z_K)
] \\
&  &
{} -
\E_{\substack{z_{i+1}, \ldots, z_K \\ \mathrm{from}~x_i(2)}}[
f_K(x_1, \ldots, x_{i-1}, x_i(2), z_{i+1}, \ldots, z_K)
] \\
& \leq &
\E_{\substack{z_{i+1}, \ldots, z_K \\ \mathrm{from}~x_i(1)}}[
g_{1,i,i+a}(x_i(1), z_{i+1}, \ldots, z_{i+a}) +
g_{1,i+a+1,i+a+b}(z_{i+a+1}, \ldots, z_{i+a+b}) \\
&  &
~~~~~~~~~~~~~~~~~~~~
{} +
f_{K - a - b - 1}(x_1, \ldots, x_{i-1}, z_{i+a+b+1}, \ldots, z_K)
] \\
&      &
{} -
\E_{\substack{z_{i+1}, \ldots, z_K \\ \mathrm{from}~x_i(2)}}[
g_{2,i,i+a}(x_i(2), z_{i+1}, \ldots, z_{i+a}) +
g_{2,i+a+1,i+a+b}(z_{i+a+1}, \ldots, z_{i+a+b}) \\
&  &
~~~~~~~~~~~~~~~~~~~~
{} +
f_{K - a - b - 1}(x_1, \ldots, x_{i-1}, z_{i+a+b+1}, \ldots, z_K)
] \\
& \leq &
c_{i,i+a} + c_{i+a+1,i+a+b} \\
&  &
{} +
\E_{\substack{z_{i+1}, \ldots, z_K \\ \mathrm{from}~x_i(1)}}[
f_{K - a - b - 1}(x_1, \ldots, x_{i-1}, z_{i+a+b+1}, \ldots, z_K)
] \\
&  &
{} -
\E_{\substack{z_{i+1}, \ldots, z_K \\ \mathrm{from}~x_i(2)}}[
f_{K - a - b - 1}(x_1, \ldots, x_{i-1}, z_{i+a+b+1}, \ldots, z_K)
] \\
& \leq &
c_{i,i+a} + c_{i+a+1,i+a+b} +
c \|p^X(1,a+b+1) - p^X(2,a+b+1)\|_1 \\
& \leq &
c_{i,i+a} + c_{i+a+1,i+a+b} + c \cdot \frac{c_{i,i+a}}{c} 
  =
2 c_{i,i+a} + c_{i+a+1,i+a+b} 
   =
d_i.
\end{eqnarray*}
This finally shows what we wanted viz.
\[
\E_{x_1, \ldots, x_K}
\left[\exp\left(\theta \sum_{i=1}^K h_i(x_1, \ldots, x_i)\right)\right]
\leq
\exp\left(\frac{\theta^2 \sum_{i=1}^K d_i^2}{8}\right).
\]

Hence,
\[
\Pr_{x_1, \ldots, x_K}[
f_K(x_1, \ldots, x_K) - \E_{z_1, \ldots, z_K}[f_K(z_1, \ldots, z_K)]
\geq \epsilon
]  \leq
e^{-\theta \epsilon} \cdot
e^{\frac{\theta^2}{8} \sum_{i=1}^K d_i^2}
\]
Setting the minimising value for 
$\theta = \frac{4\epsilon}{\sum_{i=1}^K d_i^2}$ proves the desired
upper bound on the upper tail probability.

This completes the proof of the theorem.
\end{proof}

We can now prove our new smooth expander matrix Chernoff bound for
trace distance. Note that the upper bound on the tail probability
does not involve the dimension of the ambient Hilbert space $M$ of the
quantum states, though it does involve the size of the classical alphabet
$X$ and a smooth conditional entropy.
\begin{theorem}[{\bf Smooth expander matrix Chernoff bound for 
trace distance}]
\label{thm:expanderconc}
Let $\delta > 0$.
Under the setting of Theorem~\ref{thm:expandersoftcoveringexpectaton},
with $p^X = \frac{1^X}{|X|}$ i.e. $p^X$ is the uniform probability
distribution on $X$,
\[
\Pr_{x_1, \ldots, x_K}
\left[
\|\frac{1}{K} \sum_{i=1}^K \rho_{x_i}^M - \rho^M\|_1 >
2 \sqrt{\epsilon} + \delta
\right] \leq 
2 \exp\left(-\frac{K \delta^2 (\log \lambda^{-1})^2}
		  {10 (\log |X| + \log K - \log \log |X|)^2}
      \right),
\]
if 
\[
\log K > 
\log |X| + \log \log |X| - \Hmin^\epsilon(X|M)_\rho + \log \epsilon^{-1}.
\]
\end{theorem}
\begin{proof}
We apply Theorem~\ref{thm:boundedexcisionconc} with 
$
f_i(x_1, \ldots, x_i) :=
\|\frac{1}{K} \sum_{j=1}^i \rho_{x_j}^M - \rho^M\|_1
$
for $1 \leq i \leq K$, and for any $1 \leq l_1 \leq l_2 \leq K$,
\[
g_{1,l_1, l_2}(x_{l_1}, \ldots, x_{l_2}) :=
\|\frac{1}{K} \sum_{j=l_1}^{l_2} \rho_{x_j}^M\|_1, ~
g_{2,l_1, l_2}(x_{l_1}, \ldots, x_{l_2}) :=
-g_{1,l_1, l_2}(x_{l_1}, \ldots, x_{l_2}),
\]
and
\[
c_{1, l_1 l_2} := \frac{l_2 - l_1 + 1}{K}, ~
c_{2, l_1 l_2} := -c_{1, l_1 l_2}, ~
c_{l_1 l_2} = 2 c_{1, l_1 l_2}, ~
c := 2.
\]

From Theorem~\ref{thm:expandersoftcoveringexpectaton} and the
assumed condition on $\log K$,
\[
\E_{x_1, \ldots, x_K}
\left[
\|\frac{1}{K} \sum_{i=1}^K \rho_{x_i}^M - \rho^M\|_1 
\right] < 2\sqrt{\epsilon}.
\]
By a simple calculation we get,
$a = \frac{\log |X|}{\log \lambda^{-1}}$,
$c_{i,i+a} = \frac{2(a+1)}{K}$,
$b = \frac{\log (K/(a+1))}{\log \lambda^{-1}}$,
$c_{i+a+1,i+a+b} = \frac{2b}{K}$.
Thus $d_i = \frac{4(a+1)}{K} + \frac{2 b}{K} < \frac{4(a+b+1)}{K}$.
Hence,
\[
\sum_{i=1}^K d_i^2 <
\frac{16 (a+b+1)^2}{K} <
\frac{16 (\log |X| + \log K  - \log \log |X|)^2}{K (\log \lambda^{-1})^2}.
\]
Applying Theorem~\ref{thm:boundedexcisionconc} now completes the
proof of the present theorem.
\end{proof}

\section{Conclusion}
\label{sec:conclusion}
In this work, we have obtained a novel fully smooth multipartite matrix
Chernoff bound for the trace distance under independent samples.
Our upper bound on the tail probability does not explicitly depend on 
the dimension
of the ambient Hilbert space of the quantum states. It only
depend on certain fully smooth R\'{e}nyi $2$-divergence quantities of
the ensemble of quantum states. Our upper bound is the right expression
required to prove strong inner bounds for private classical communication 
over multiterminal quantum wiretap channels with many non-interacting
eavesdroppers; the resulting inner bounds are independent of the
dimensions of eavesdroppers' Hilbert spaces. Such powerful inner bounds
for wiretap channels were unknown even in the fully classical setting.
We also prove the first smooth expander matrix Chernoff bound for the
trace distance. Again, the upper bound on the tail probability does not
explicitly depend on the dimension
of the ambient Hilbert space of the quantum states, 
though it does involve the size of the classical alphabet
and a smooth conditional entropy of the ensemble.
Our new expander matrix Chernoff bound is proved by generalising 
McDiarmid's method
of bounded differences for functions of independent random variables
to a new method of bounded excision for functions of
expander walks.

An immediate open problem is to generalise our concentration method
via bounded excision for expander walks to the multipartite setting.
The main bottleneck in this endeavour is that we have no good way
of bounding the expected value of the function under expander walks
in the multipartite setting, because expander walks do not satisfy
pairwise independence amongst the samples. The fully smooth multipartite 
soft covering lemma in expectation of \cite{Sen:flatten} does hold 
when pairwise
independence amongst the samples is `slightly broken'. Unfortunately,
the notion of `slightly broken' in that work fails to capture expander
walks. Extending the methods of \cite{Sen:flatten} to handle expander
walks will be a good challenge.

\section*{Acknowledgements}
I acknowledge support of the Department of Atomic Energy, Government
of India, under project no. RTI4001.

\bibliography{matrixchernoff}

\begin{thebibliography}{GLSS18}

\bibitem[ADJ17]{Jain:convexsplit}
A.~Anshu, V.~Devabathini, and R.~Jain.
\newblock Quantum communication using coherent rejection sampling.
\newblock {\em Physical Review Letters}, 119:120506:1--120506:21, 2017.

\bibitem[AJW18]{anshu:slepianwolf}
A.~Anshu, R.~Jain, and N.~Warsi.
\newblock A generalized quantum {Slepian–Wolf}.
\newblock {\em IEEE Transactions on Information Theory}, 64:1436--1453, 2018.

\bibitem[AW02]{Ahlswede:matrixchernoff}
R.~Ahlswede and A.~Winter.
\newblock Strong converse for identification via quantum channels.
\newblock {\em IEEE Transactions on Information Theory}, 48(3):569--579, 2002.

\bibitem[Che52]{Chernoff:conc}
H.~Chernoff.
\newblock A measure of asymptotic efficiency for tests of a hypothesis based on
  the sum of observations.
\newblock {\em Annals of Mathematical Statistics}, 23:493--509, 1952.

\bibitem[CNS21]{Chakraborty:wiretapQMAC}
S.~Chakraborty, A.~Nema, and P.~Sen.
\newblock One-shot inner bounds for sending private classical information over
  a quantum mac.
\newblock In {\em IEEE Information Theory Workshop (ITW)}, 2021.
\newblock Full version in arXiv::2105.06100.

\bibitem[CPS22]{Chakraborty:measurementcompression}
S.~Chakraborty, A.~Padakandla, and P.~Sen.
\newblock Centralised multi link measurement compression with side information.
\newblock In {\em IEEE International Symposium on Information Theory (ISIT)},
  2022.
\newblock Full version in arXiv::2203.16157.

\bibitem[Cuf16]{Cuff:softcovering}
P.~Cuff.
\newblock Soft covering with high probability.
\newblock In {\em IEEE International Symposium on Information Theory (ISIT)},
  2016.
\newblock Full version at arXiv::1605.06396.

\bibitem[Dev05]{Devetak:wiretap}
I.~Devetak.
\newblock The private classical capacity and quantum capacity of a quantum
  channel.
\newblock {\em IEEE Transactions on Information Theory}, 51:44--55, 2005.

\bibitem[Gil98]{Gillman:expanderchernoff}
D.~Gillman.
\newblock A {Chernoﬀ} bound for random walks on expander graphs.
\newblock {\em SIAM Journal on Computing}, 27(4):1203--1220, 1998.

\bibitem[GLSS18]{Garg:expandermatrixchernoff}
A.~Garg, Y.~Lee, Z.~Song, and N.~Srivastava.
\newblock A matrix expander {Chernoff} bound.
\newblock In {\em ACM SIGACT Symposium on Theory of Computing (STOC)}, pages
  1102--1114, 2018.
\newblock Also arXiv::1704.03864 (math.PR).

\bibitem[Hoe63]{Hoeffding:conc}
W.~Hoeffding.
\newblock Probability inequalities for sums of bounded random variables.
\newblock {\em Journal of the American Statistical Association}, 58:13--30,
  1963.

\bibitem[McD89]{McDiarmid:bounded}
C.~McDiarmid.
\newblock On the method of bounded differences.
\newblock In {\em Surveys in Combinatorics}. Cambridge University Press, 1989.

\bibitem[MR95]{MotwaniRaghavan}
R.~Motwani and P.~Raghavan.
\newblock {\em Randomized Algorithms}.
\newblock Cambridge University Press, 1995.

\bibitem[RSW17]{Radhakrishnan:wiretap}
J.~Radhakrishnan, P.~Sen, and N.~Warsi.
\newblock One-shot private classical capacity of quantum wiretap channel:
  {Based} on one-shot quantum covering lemma.
\newblock Proceedings of QCRYPT workshop. Also arXiv:1703.01932, 2017.

\bibitem[Sen24a]{Sen:telescoping}
P.~Sen.
\newblock Fully smooth one shot multipartite covering and decoupling of quantum
  states via telescoping.
\newblock arXiv:2410.17893, 2024.

\bibitem[Sen24b]{Sen:flatten}
P.~Sen.
\newblock Fully smooth one shot multipartite soft covering of quantum states
  without pairwise independence.
\newblock arXiv:2410.17908, 2024.

\bibitem[Tro15]{Tropp:book}
J.~Tropp.
\newblock An introduction to matrix concentration inequalities.
\newblock {\em Foundations and Trends\,{\rm \sffamily\textregistered} in
  Machine Learning}, 8(1-2):1--230, 2015.

\bibitem[Wil17]{Wilde:wiretap}
M.~Wilde.
\newblock Position-based coding and convex splitting for private communication
  over quantum channels.
\newblock {\em Quantum Information Processing}, 16:264:1--264:31, 2017.

\bibitem[Win04]{Winter:measurementcompression}
A.~Winter.
\newblock {``Extrinsic''} and ``intrinsic'' data in quantum measurements:
  {Asymptotic} convex decomposition of positive operator valued measures.
\newblock {\em Communications in Mathematical Physics}, 244(1):157--185, 2004.

\bibitem[WX05]{Xiao:expandermatrixchernoff}
A.~Wigderson and D.~Xiao.
\newblock A randomness-efficient sampler for matrix-valued functions and
  applications.
\newblock In {\em IEEE Symposium on Foundations of Computer Science (FOCS)},
  pages 397--406, 2005.

\bibitem[Yin04]{Ying:bounded}
Y.~Ying.
\newblock {McDiarmid’s} inequalities of {Bernstein} and {Bennett} forms.
\newblock Technical report, City University of Hong Kong. Available at
  https://empslocal.ex.ac.uk/people/staff/yy267/McDiarmid.pdf, 2004.

\end{thebibliography}

\end{document}